\documentclass[pra,twocolumn]{revtex4-2}%
\usepackage[dvips]{graphicx}%
\usepackage{bm,color}
\usepackage{amsmath,amssymb,braket}
\usepackage{amsthm}

\usepackage[dvipdfmx]{hyperref}
\hypersetup{% hyperref
setpagesize=false,
 bookmarksnumbered=true,%
 bookmarksopen=true,%
 colorlinks=true,%
 linkcolor= hyptxt, %mygray, %magenta, %blue,
 citecolor=citetxt,
} % https://www.isc.meiji.ac.jp/~mizutani/tex/link_slide/hyperlink.html

\usepackage{xcolor}
\definecolor{hyptxt}{rgb}{0,0.5,1}
\definecolor{citetxt}{rgb}{0,0.5,1}
\definecolor{mygray}{gray}{0.5}

\newcommand{\ketbra}[2]{\ket{#1}\!\bra{#2}} %\newcommand{\ketbra}[3]{\ket{#1}_{#3}\bra{#2}}

\theoremstyle{definition}
\newtheorem{theorem}{Theorem}%[section]  https://mathlandscape.com/latex-amsthm/
\newtheorem{corollary}[theorem]{Corollary}
\newtheorem{remark}[theorem]{Remark}
\begin{document}
%%%%%%%%
%  \small
%%%%%%%%
\title{Strong convergence of a resolution of the identity via canonical coherent states}

\author{Ryo Namiki}%
\affiliation{Department of Physics, Gakushuin University, 1-5-1 Mejiro, Toshima-ku, Tokyo 171-8588, Japan}
%\affiliation{Department of Physics, Graduate School of Science, Kyoto University, Kyoto 606-8502, Japan}
%\affiliation{Institute for Quantum Computing and Department of Physics and Astronomy, University of Waterloo, Waterloo, Ontario, N2L 3G1, Canada}

%\date{\today}%{September 13, 2010}%\today}
\begin{abstract} 
A resolution of the identity due to canonical coherent states is often proven in the weak operator topology. However, such a resolution with an integral symbol is typically supposed to hold in the strong operator topology associated with the framework of the spectral theorem. We provide an elementary proof of the strong convergence for the resolution of the identity due to canonical coherent states starting with a mostly familiar setup. Further, we enjoy a different proof and show that the relevant uniform limit does not exist.
\end{abstract}

% insert suggested PACS numbers in braces on next line
% \pacs{03.67.Dd, 42.50.Lc} 
% insert suggested keywords - APS authors don't need to do this
%\keywords{quantum cryptogrphy}
\maketitle

%\newpage
%%%%%%%%%%%%%%%%%%%%%   part 1 %%%%%%%%%%%%%%%%%%%%%%%%%%%%%%%%%%%%%%%%%%%%%%%%
%%%%%%%%%%%%%%%%%%

%(\today) 

%%%%%%%%%%%%%%%%%%%%%%
\section{Introduction}\label{IntroSC}
%%%%%%%%%%%%%%%%%%%%%%

Canonical coherent states play significant roles in cultivating quantum optics and other wide area of quantum physics \cite{CSReview95,Antoine2000,GarciadeLeon08,Parisio2010}.
While they are certainly useful in some operator calculus,
it is often noted in literatures that the operator identity 
\begin{align}
I %\overset {s} {=}
=    \pi ^{-1} \int \ket{\alpha}\bra{\alpha} d ^2 \alpha \label{ResI00}
\end{align} holds in the weak sense. This means that an arbitrary inner product  satisfies the following relation: %are given by 
\begin{align}
\braket{\psi | \phi} = \pi ^{-1} \int \braket{\psi | \alpha }\braket{\alpha| \phi} d^2 \alpha. \label{Qn221019-1}
\end{align}
%This formula is certainly helpful to calculate the inner product. Moreover, 
Defining operator's action through this form could be convenient in the point that one can avoid to define an operator-valued integral, however, this seems to have no other merits. Although an operator on finite dimensional linear algebra is completely characterized by its matrix elements, 
 such characterization is not generally hold for operators regarding infinite dimensional systems.  In fact,  textbooks of operator algebra states that the convergence in the weak topology does not generally imply the convergence in the strong topology. Therefore, the weak convergence of Eq.~\eqref{Qn221019-1} does not immediately admit the decomposition in the strong sense:
\begin{align}
\phi = \pi ^{-1}\int \ket{\alpha}\braket{\alpha | \phi} d ^2 \alpha.   \label{expansionSZeo}
\end{align} 
Here,  the form of the integral that takes a value on a Hilbert space  (the $\cal H$-valued integral) reminds us the spectral decomposition theorem and, it would be somehow surprising if such a familiar expansion is thought to be unavailable. In the spectral theorem, a resolution of the identity is composed of a family of monotone projection operators and an integral form of the operator decomposition holds in the strong  operator topology. We may expect the same convergence topology for the family of coherent states because  it can define a monotone sequence of positive operators though not projective. 
In contrast,  it is typical to find a statement which notices the integral operator holds in the  weak sense in literatures, and there seems almost no chance to encounter a complete statement of the strong convergence. % holds in strong operator topology. 
  Let alone the definition of  the $\cal H$-valued integral.

Consequently,  we may frequently use the decomposition due to the strong convergence Eq.~\eqref{expansionSZeo} regardless there is almost no way to find its formal proof in textbooks. 
As an interesting exception, one can find a comment in one of {Klauder's lecture note} %Klauder's lecture notes 
that states the strong convergence with an  outline of a possible proof \cite{KlauderLectureNote}. Unfortunately, it appears as a part of an exercise, and is less likely to be spotted.  %\href{http://www.phys.ufl.edu/~klauder/norway/}{Klauder's lecture note}

In passing it could be natural to ask why its norm convergence is unavailable in the first place. Similar to an integral of a constant function over  an infinite volume, it is almost trivial that  the operator-valued integral in Eq.~\eqref{ResI00} does not exists, and thus the resolution of the identity holds at most in the strong sense. % as it is often emphasized in the spectral theorem. 
Again, it seems unlikely to encounter a proof of the nonexistence. %Overall, there are rather few chances to concern the

To this end, it would be worth making a definite statement whether or not the resolution of the identity in Eq.~\eqref{ResI00} holds in the strong sense. In addition, it seems better to confirm the nonexistence of the  resolution of the identity with regard to the norm topology once in a while.

In this article, we provide an elementary proof  of the strong convergence of the resolution of the identity via canonical coherent states with a mostly familiar setup.  We start with basic definitions  and prove its convergence with regard to  the strong operator topology in Sect.~\ref{SectBesicProof}. %, 
    We review  Klauder's approach as a different proof that lifts up the weak convergence to the strong convergence in Sect.~\ref{SectKlauderApproach}. We give a formal proof that the uniform limit does not exist in Sect.~\ref{SecNoUniformLimit}.    We conclude this article in Sect.~\ref{ClosingSECTION}.

%#### Problems for: selected topics of quantum theory  by John Klauder 

%#### Topic 1. 1.5  [pdf](http://www.phys.ufl.edu/~klauder/norway/norwayproblems.

\section{Basic notions and an elementary proof of the strong convergence} \label{SectBesicProof}
Let $(\ket{n})_{n=0}^\infty $ be an orthonormal basis and $ {\cal H} = \ell ^2[0,\infty)$ be a complex hilbert space spanned by  $(\ket{n})_{n=0}^\infty $. 
Any vector $\phi \in \cal H$ admits the orthonormal expansion 
\begin{align}\label{abState}
\phi = \sum_{n=0}^\infty a_n \ket{n}  
\end{align}
 with $\sum_{n=0}^\infty |a_n|^2 <\infty$. 
 We will define the inner product by $\braket{\phi| \varphi} := \sum_{n=0}^\infty  a_n ^* b_n$ for $\phi = \sum_{n=0}^\infty a_n \ket{n}  $ and $\varphi = \sum_{n=0}^\infty b_n \ket{n} $. The norm on $\cal H$  is defined as \begin{align}  \| \phi \| = \sqrt {\braket{\phi |\phi}}. \end{align}
A coherent state with an amplitude $\alpha \in \mathbb C$, is defined as 
\begin{align}
 \ket{\alpha } = e^{- |\alpha |^2 /2} \sum_{n=0}^\infty \alpha ^n \ket{n} / \sqrt {n!}. \label{coState}
\end{align}
It holds $\|\ket{\alpha}\| =1$ and $\braket{ n |\alpha}= e^{- |\alpha |^2 /2} \alpha ^n / \sqrt {n!} $. 

 We may concern the following three topologies of operator's convergence. 
We say an operator sequence $(A_n)$ converges to $A$ in the weak operator topology if
\begin{align}
A =  w\text{-}\lim_{n \to \infty} A_n \ \Leftrightarrow \    
\lim_{n \to \infty}  \braket{ \psi | A- A_n | \phi} =0, \ \forall \phi,\psi \in \cal H.  \end{align}  
We say $(A_n)$ converges to $A$ in the strong operator topology  (or the topology of $\cal H$) if 
\begin{align}
A =  s\text{-}\lim_{n \to \infty} A_n \ \Leftrightarrow \    
\lim_{n \to \infty} \|A \phi - A_n \phi \| =0, \  \forall \phi \in \cal H.
 \end{align}% where the norm in  $\ell ^1$ topology is defined as \begin{align}  \| \phi \| = \sqrt {\braket{\phi |\phi}}. \end{align}
%In this paper, we refers to the norm as this $\ell ^1$ norm. 
We say $(A_n)$ converges to $A$ in the uniform operator topology (or norm topology) if 
\begin{align}
A =  \lim_{n \to \infty} A_n \quad \Leftrightarrow \quad   \label{defNormTop}
\lim_{n \to \infty} \|A   - A_n  \| =0,   \end{align}  
where the operator norm is defined by $\sup_{\| \phi \| \le 1 }\| A \phi\| $.

Our primary goal %of the strong convergence for the resoluton of the identy via coherent states 
is to show the strong convergence: 
\begin{align}
I = s \text{-} \lim_{n \to \infty} \left( \int _{|\alpha | \le n} \ket{\alpha } \bra{\alpha}  \frac{  d^2 \alpha  }{\pi} \right) .
\end{align}
This can be accomplished by the following theorem, and we can safely use the decomposition in Eq.~\eqref{expansionSZeo}.

\begin{theorem} \label{Theorem10}
Let be $\varphi \in {\cal H}$. For any $\epsilon >0$, there exists $R>0 $ such that 
%$(\forall \epsilon >0)(\forall \varphi \in  {\cal H} ) \exist R >0 ;(\forall  r \ge R)$
\begin{align}
\left \|  \ket{\varphi } -\left( \int _{|\alpha | \le r} \ket{\alpha } \braket{\alpha| \varphi}  \frac{  d^2 \alpha  }{\pi} \right)  \right \| < \epsilon
\end{align}
whenever  $r  \ge R$. % Note that $\int _{|\alpha | \le r} \ket{\alpha } \braket{\alpha| \varphi}   {  d^2 \alpha  }/{\pi}  \in {\cal H}$ is guaranteed by the finiteness of the integration range. 
\end{theorem}
\begin{remark} \label{remark20} Since $\braket{\alpha| \varphi}$ is a uniformly  bounded  continuous function of $\alpha$, and the state vector $ \ket{\alpha}$ is continuous in the sense $ \| \ket{\alpha} -\ket{\beta }\| \to 0 \quad (|\alpha -\beta | \to 0 )$, the integrand $ \ket{\alpha } \braket{\alpha | \varphi}$ is continuous and norm bounded on $| \alpha | \le R $. Therefore, the vector-valued Riemann sum over  the finite area $| \alpha | \le R $ converges to a state vector in $\cal H$. This gives a $\cal H$-valued integral and guarantees $\int _{|\alpha | \le r} \ket{\alpha } \braket{\alpha| \varphi}   {d^2 \alpha  }/{\pi}  \in {\cal H}$. 
We can deal with the integrability beyond the continuous functions  in terms of the Bochner integral. 
Notably, a ${\cal H}$-valued function is integral iff its norm is square-integrable   (See Theorem~\ref{NFBochner} in  Appendix~\ref{B-integral}).
\end{remark}

%%%%%%%%%%%%%%%%%%%%%%%%%
\begin{proof}
%%%%%%%%%%%%%%%%%%%%%%%%% 
%\subsection{propeties of $I_n$}
Let be $r >0$, and let us  % and $n \in \{0,1,2, \dots \}$. 
 define 
\begin{align}
I_n(r) = \int ^{r}_{0}\dfrac{y^{n}e^{-y}}{n!}dy, \qquad \left( n =0,1,2,3,\cdots \right). 
\end{align}
We will repeatedly use the following properties (See  Appendix~\ref{AppIntegral1} for a proof). 
%\\ (i) $I_{n+1}(R) \le I_n(R)  $ \\ (ii) $0\le I_n(R) \le 1$  \\ (iii) $\lim_{R \to \infty} I_n(R) =1 $. 
\begin{align}
 \text{(i) }& I_{n+1}(r) \le I_n(r), % \nonumber\\
 \text{ (ii) }  0\le I_n(r) \le 1, \nonumber\\ 
 \text{ (iii) }& \lim_{r \to \infty} I_n(r) =1, \text{ (iv) }  |1- I_n (r)| \le 1 .  \label{Iprop}
\end{align}
The property (iv) follows from (i) and (ii).

From the expansion in Eqs.~\eqref{abState} and \eqref{coState}, and a  somewhat lengthy process (see Appendix \ref{somewhatlengthy}),  we have 
\begin{align} & \left \|  \ket{\varphi } -\left( \int _{|\alpha | \le r} \ket{\alpha } \braket{\alpha| \varphi}  \frac{  d^2 \alpha  }{\pi} \right)  \right \|^2 \nonumber \\ %\left \|  \ket{\varphi } -\left( \int _{|\alpha | \le R} \frac{ \ket{\alpha } \bra{\alpha}}{\pi} d^2 \alpha  \right)  \ket{\varphi} \right \| ^2  \nonumber \\
= &\left\| \sum _{n}\left( 1-\int ^{r^{2}}_{0}\dfrac{y^{n}e^{-y}}{n!}dy\right)   \varphi _{n} \ket{n}\right\| ^2  \nonumber\\
 = &  \sum _{n = 0}^\infty\left( 1- I_n(r^2)  \right) ^{2}   | \varphi _{n} |^2  \nonumber \\
 \le & \sum _{n = 0}^ k\left( 1-I_n(r^2) \right) ^{2}   | \varphi _{n} |^2  +  \sum _{n = k+1 }^\infty  | \varphi _{n} |^2 \label{Int6}
\end{align}
where we used the property (iv) in Eq.~\eqref{Iprop} 
to obtain the last inequality.

Let be $\epsilon >0$. 
Since $\varphi \in \cal H$ we can select a sufficiently large  $K \in \mathbb N$ such that it holds for $k \ge K$%so as to  satisfy
\begin{align}
\sum _{n = k+1 }^\infty  | \varphi _{n} |^2 < \frac{1}{2} \epsilon.  \label{Int7}
\end{align}
From the properties (i) and (ii) in Eq.~\eqref{Iprop},  $n \le k$ implies   
\begin{align}
\left( 1-I_n (r^2)\right)^2 \le  (1- I_k(r^2))^2.   
\end{align}
This relation leads to  \begin{align}
& \sum _{n = 0}^ k\left( 1-I_n(r^2)\right) ^{2}   | \varphi _{n} |^2 \nonumber \\
\le& \left( 1- I_k(r^2) \right)^2 \sum _{n = 0}^ k   |\varphi_n |^2 %\nonumber \\
\le   \left( 1- I_k(r^2) \right) ^2 \| \varphi \|^2. 
 \end{align}
From the property (iii) in Eq.~\eqref{Iprop}, 
%Here recalling $I_k(R) \to 1 \quad (R\to \infty)$, 
we can select a sufficiently large $R>0$  such that, for $r \ge R$, it holds 
\begin{align}
\left( 1- I_k(r^2) \right) ^2  \|\varphi\| ^2 < \frac{1}{2}\epsilon. \label{Int9}
\end{align}

Concatenating Eqs.~\eqref{Int6},~\eqref{Int7}, and~\eqref{Int9}
we obtain \begin{align} & \left \|  \ket{\varphi } -\left( \int _{|\alpha | \le r} \frac{ \ket{\alpha } \bra{\alpha}}{\pi} d^2 \alpha  \right)  \ket{\varphi} \right \| ^2  \nonumber \\ 
 \le & (1+I_k(r))^2 \| \varphi \|^2 %\sum _{n = 0}^ k\left( 1-\int ^{R^{2}}_{0}\dfrac{y^{n}e^{-y}}{n!}dy\right) ^{2}   | \varphi _{n} |^2 
  +  \sum _{n = k+1 }^\infty  | \varphi _{n} |^2 <  \epsilon.   %\|\ket{\varphi}\|^2
\end{align}
This proves the statement of Theorem~\ref{Theorem10} and  concludes the strong convergence for the resolution of the identity due to coherent states:    
\begin{align}
 I = s\text{-}\lim_{r \to \infty} \int _{|\alpha | \le r} \frac{ \ket{\alpha } \bra{\alpha}}{\pi} d^2 \alpha.  
\end{align}

%%%%%%%%%%%%%%%%%%%%%%%%%
\end{proof}
%%%%%%%%%%%%%%%%%%%%%%%%% 

 \begin{remark} \label{remark30}For our primary purpose, the $\cal H$-valued integral as in Eq.~\eqref{QnstAn} is sufficient, and 
  it is unnecessary to define an operator valued integral such as 
   \begin{align}
A_n =  \int_{ |\alpha  | \le n} \ket{\alpha}\bra{\alpha} \frac{d^2\alpha}{\pi} \quad (n =1,2,3,\dots). \label{AnOpVal}
\end{align} 
 However, it would be worth noting that this operator-valued integral does exist as long as $n$ is finite. We can proof this fact similar to Remark~\ref{remark20}. Since the density operator of a coherent state $ \ketbra{\alpha}{\alpha}$ is continuous in the sense $ \| \ketbra{\alpha}{\alpha} -\ketbra{\beta }{\beta}\| \to 0 \quad (|\alpha -\beta | \to 0 )$,  % $ \ket{\alpha } \braket{\alpha | \varphi}$ is continuous and norm bounded on $| \alpha | \le R $. 
%Therefore, 
the operator-valued Riemann sum over  a finite area $| \alpha | \le R $ converges to a compact operator. The integral also converges in the trace norm topology as long as the integration volume is finite. 
\end{remark}

\begin{remark}
As we will show in Sect.~\ref{SecNoUniformLimit}, %We note that 
$(A_n)$ of Eq.~\eqref{AnOpVal}, does not converge to the unit operator $I$ in the norm topology of Eq.~\eqref{defNormTop}, and this sequence has no uniform limit. %In fact, $(A_n)$ cannot be a Cauchy sequence in  operator norm and the limit in norm topology does not exist. 
Such a property can be seen  on a sequence of projection operators in the form $ B_n := \sum_{k=0}^n \ket{k}\bra{k} $. In fact, this sequence $(B_n)$ does not converge to the unit operator $I$ in the operator norm topology. Moreover, since $\| B_n -B_m\| =1 \quad (n \neq m)$, no subsequence of $ (B_n)$ converges in the norm topology. A norm space  is referred to as the compact space when any bounded sequence has a convergent subsequence. In this regards, %the sequence of $(B_n)$ shows that 
the space of bounded operators is not compact, and even quite a simple decomposition such as $I = \sum_{n} \ket{n}\bra{n}$ is unavailable with respect to the norm topology unless the dimension is finite. 
\end{remark}

%%%%%%%%%%%%%%%%%%%%%%%%%%%%%%%%%%%%%%%%%%%
\section{Review of Klauder's approach} \label{SectKlauderApproach}%{Method in operator algebra}
%%%%%%%%%%%%%%%%%%%%%%%%%%%%%%%%%%%%%%%%%%%
Here we assume the weak convergence and prove the strong convergence %in terms of abstruct operator algebra as
 based on the outline given in Klauder's lecture note %\href{http://www.phys.ufl.edu/~klauder/norway/}{Klauder's lecture note}
  \cite{KlauderLectureNote}. 
 
 We consider a sequence of positive operators defined as
\begin{align}
 A_n \ket{\varphi}:=  \int_{ |\alpha  | \le n} \ket{\alpha}\braket{\alpha| \varphi} \frac{d^2\alpha}{\pi} \quad (n =1,2,3,\dots). \label{QnstAn}
\end{align} The existence of this $\cal H$-valued integral is guaranteed by the prescription noted in Remark~\ref{remark20}. 
From the construction it holds  $\braket{\varphi| A_n | \varphi} \ge 0$ (See Appendix~\ref{BraAction}), and thus $A_n$ is positive. In what follows, we denote this operator positivity by $A_n\ge 0$.   Since  $ A_n \ket{\varphi} = \sum_{k=0}^\infty I_k(n^2) \varphi_k \ket k $ holds, we can confirm
\begin{align}
\braket{\varphi|A_n | \varphi} \le \| \varphi \|^2 , \quad\braket{\varphi|A_n | \varphi} \le \braket{\varphi|A_{n+1}|\varphi}, \label{Apositive} \\
\|A_n \varphi\| \le (I_0(n^2))^2\|\varphi\| \le \|\varphi\|, \label{Abounded}
\end{align}
where we repeatedly use the propery (i)~and~(ii) in Eq.~\eqref{Iprop}. %of $I_{n+1}(R) \le I_n(R)  $ and $0\le I_n(R) \le 1$ (See appendix~\ref{AppIntegral1}).  

The relations in Eq.~\eqref{Apositive} imply $A_n \le I$ and  $A_{n} \le A_{n+1}$ for $n \in \{0, 1,2, \dots \}$.  The relation in Eq.~\eqref{Abounded} implies $(A_n ) $ is a sequence of bounded operators and their operator norm is bounded as $\|A_n\| \le 1 = \| I\| $.  

Up to here, we have confirmed that $(A_n)$ is a sequence of positive bounded operators which satisfies  (i) $ 0 \le A_n \le I$ and (ii) $A_n \le A_{n+1}$. %In the sequel, we assume the weak convergence $ \lim _{n \to \infty }\braket{ \psi |A_n | \phi} =\braket{\psi |\phi}$ and show that $A_n$ converges to $I$ in strong operator topology, namely  \begin{align} \lim_{n \to \infty }\| (I- A_n )\phi \| =\lim_{n \to \infty }\| \phi - A_n \phi \| =0 , \ \forall \phi \in \cal H.  \end{align}
Armed with this boundedness and monotonicity, in his lecture note, Klauder suggested to prove the following theorem:
\begin{theorem}
Let $0 \le A_n \le I$ and $A_n \le A_{n+1}$. Suppose that $(A_n)$ converges to $I$ in the weak operator topology as 
\begin{align}
 \lim_{n \to \infty} \braket{\phi|I -A_n|\varphi}=0 \quad \forall \phi,\varphi \in{\cal H}. 
\end{align}
Then, $(A_n) $ converges to $I$ in the strong operator topology, namely 
\begin{align}
\lim_{n \to \infty} \left\| \left( I-A_{n}\right) \phi \right\| =0  \quad \forall \phi \in{\cal H}. 
\end{align}
\end{theorem}

\begin{proof}
Let be $\phi \in \cal H$. A straightforward calculation leads to %, we have
\begin{align}
\left\| \left( I-A_{n}\right) \phi \right\|^2 =& \langle \left( I-A_{n}\right) \phi | \left( I-A_{n}\right) \phi \rangle \nonumber \\
= &\left\| \phi \right\| ^{2}-\left \langle  A_{n}\phi | \phi \right \rangle -\langle \phi | A_{n}\phi \rangle +\left\| A_{n}\phi \right\| ^{2}.   \nonumber
\end{align}  From this formula and  the weak limit,  $\left \langle  A_{n}\phi |\phi \right \rangle  \to \|\phi\|^2$
and 
 $\left \langle  \phi | A_{n}\phi \right \rangle  \to \|\phi\|^2$, we only have to prove the following convergence:  
\begin{align}
 \|A_n \phi\| \to \|\phi\| \quad (n \to \infty) . 
\end{align}

 Let us admit that a positive operator has a unique positive square root. From the decomposition $(A_n-A_n^2) = A_n^{1/2}(I-A_n)A_n^{1/2}$   and $I-A_n \ge 0$, it holds  $A_n- A_n^2 \ge0 $, i.e.,  $A_n \ge A_n^2$. Then, $I \ge A_n \ge A_n^2$ leads to  $\|\phi\|^2\ge \braket{\phi |A_n |\phi} \ge \braket{\phi|A_n^2|\phi}=\|A_n \phi\| ^2$. We thus have 
 \begin{align}
 \|\phi\| \ge \|A_n \phi\|. \label{SONO1}
\end{align}
 In turn, Schwarz's inequality yeilds $\braket{\phi | A_n\phi} \le \|\phi\| \| A_n\phi\|$. The weak limit of the left-hand-side term %side of this inequaltion
   implies 
   \begin{align}
 \|\phi\|^2=\lim_n \braket{\phi,A_n \phi} \le \|\phi\| \lim_n \|A_n \phi\| . %\label{SONO2}
\end{align} 
 This relation together with Eq.~\eqref{SONO1} % and \eqref{SONO2} 
 lead to  
 \begin{align}
 \| \phi\| \le \lim_n\| A_n \phi\| \le \| \phi\|. 
\end{align} We thus conclude $ \lim_n\| A_n \phi\| =\| \phi\|$.
\end{proof}

%%%%%%%%%%%%%%%%%%%%%%%%%%%
\section{no uniform limit} \label{SecNoUniformLimit}
%%%%%%%%%%%%%%%%%%%%%%%%%%%
In this section, we quickly prove that the sequence of operators $(A_n)$ does not converge uniformly.
\begin{theorem} \label{TheoNonUni}
The following relation holds for the operator sequence $(A_n)$ defined as in Eq.~\eqref{QnstAn}:  % does not converge uniformly to the identity operator $I$ as it holds   
\begin{align}
 (\forall \epsilon >0) (\forall n  \in \mathbb N ) \quad\|A_n-I\|> 1-\epsilon .  \label{nonUniform000}
\end{align}
 \end{theorem}

\begin{remark}
The relation in Eq.~\eqref{nonUniform000} obviously makes a contradiction to the statement of the uniform convergence:
\begin{align}
 (\forall\epsilon >0 )\exists N>0;(n>N) \quad\| A_n-I\|< \epsilon. 
\end{align} 
Hence, this theorem implies the operator sequence $(A_n)$ has no uniform limit.
\end{remark}

\begin{proof}
We can readly estimate
  \begin{align}
 \|A_n -I\|^2  \ge&\|(A_n-I)\ket{m}\|^2 \nonumber \\
 =&\|A_n \ket{m}\|^2+1 -2\braket{m|A_n|m } \nonumber \\
 >& 1 -2\braket{m|A_n|m } \label{20230321QnUniform}
.\end{align}
We would like to show that the expression in the last line approaches to $1$ when $m  \to \infty $. Using Eq.~\eqref{coState}, we have
\begin{align}
 \braket{m|A_n|m }=& \int _{|\alpha| \le n } \frac{|\alpha|^{2m}  e^{-|\alpha ^2|} }{m!} \frac{d^2 \alpha}{\pi} \nonumber \\
 =&\frac {2 }{m!}\int_0^n  {r^{2m+1} e^{-r^2}}\ dr \le \frac {2 n^{2m+2}}{m!},
\end{align} where the last inequality comes from $e^{-r^2} \le 1$ for $r\in [0,\infty)$.
This implies 
\begin{align}
 \braket{m|A_n|m }\to 0 \quad( m\to \infty). 
\end{align}
This relation with Eq.~\eqref{20230321QnUniform} concludes the statement of our theorem~\ref{TheoNonUni}. 
\end{proof}

\begin{remark} In the spectral theory, it is typical to consider the monotone family of projection operators. As we have mentioned in the introduction (Sect.~\ref{IntroSC}), the family of coherent states lives outside of such a family. The present theorem articulates   it is  the monotonicity and boundedness that leads to the strong convergence. % Hence, this theorem could be useful to show the strong convergence for the resolutions due to an over complete bases. % would converge in the strong operator topology.   
\end{remark}%Note that the norm convergence 

\begin{remark}%Note that the norm convergence 
We may frequently use the resolution of the identity by the sequence due to an orthonormal base, 
\begin{align}
I = s\text{-}\lim_{n \to \infty } \sum_{k=0}^n \ket{k}\bra{k}. %\| B_n -I \| =1 \qquad n \in \mathbb N.
\end{align}
This sequence also fails to converge in the norm topology unless the dimension is finite. 
In fact, as a counterpart of Eq.~\eqref{nonUniform000}, it holds %$\| B_n -I \| =1 \qquad n \in \mathbb N$. 
\begin{align}
\left \| \sum_{k=0}^n \ket{k}\bra{k} -I \right \| =1 \qquad ( n \in \mathbb N).
\end{align} Here, even a countable summation is unable to  shorten the norm distance from the identity operator. % and the difference from the identity operator never be smaller in the norm topology.  
Because of this structure, the uniform limit is unavailable. So we expect the strong convergence, at most.%If we consider the Bochner integral for the space of bounded operators, unboundedness of the integral of the norm of the operator-valued functions immediately show us the nonexistence%In our case, the point for the norm convergence is the finiteness of the integration volume as we have seen in Remark~\ref{remark30}.% point for the the same as   In this respect, one may normally consider 
\end{remark}

\section{conclusion and remarks} \label{ClosingSECTION}% {summary} %
In quantum optics, an integral decomposition of the identity operator due to a family of coherent states is used as a standard theoretical tool. This  decomposition is often introduced  with the remark describing its weak convergence. This could mislead readers into interpreting that  the strong convergence is unavailable. If the strong convergence is available there is rather no point to mention its weak convergence. This is because the convergence  in the strong topology  automatically guarantees the convergence in the weak topology.  In this circumstance,  it would be worth exposing available convergence topologies for the resolution of the identity via the family of coherent states. 

In this article, we have proven  that the integral form  resolution of the identity due to canonical coherent states holds in the strong operator topology in a mostly familiar elementary setting (Sect.~\ref{SectBesicProof}). We have also given a different proof based on Klauder's lecture notes (Sect.~\ref{SectKlauderApproach}). This proof lifts up the weak convergence to the strong one by the monotonicity and boundedness in more abstract operator algebraic taste. We further have shown that the corresponding uniform limit is unavailable (Sect.~\ref{SecNoUniformLimit}). 

Therefore, one can safely use the integral-form operator decomposition Eq.~\eqref{coState}  in the strong sense similar to the spectral theorem, and claim that the operator-valued integral holds in the strong topology. In turn no uniform limit is available and the solo integral form is meaningless.  It should be mandatory to check  which convergence topology is available  together with a properly definition of the vector-valued integral when such an operator decomposition appears.

 The details of our proof could be much simpler and shorter if one admits dominated convergence theorem for the Bochner integral. We thus have prescribed some detail of the Bochner integral  in Appendix~\ref{B-integral}  so that  readers familiar with the Lebesgue integral would be almost immediate to prove dominated convergence theorem for the Bochner integrals.

We have refrained from mentioning the frame theory \cite{Christensen2001}. The set of canonical coherent states forms a tight frame and its strong convergence might be obvious. We hope this viewpoint also help us to spread the concise statement of the strong convergence. 

%%%%%%%%%%%%%%%%%%%%%%%%%%%
%%%%%%%%%%%%%%%%%%%%%%%%%%%
\appendix
%%%%%%%%%%%%%%%%%%%%%%%%%%%
%%%%%%%%%%%%%%%%%%%%%%%%%%%

%%%%%%%%%%%%%%%%%%%%%%%%%%
\section{Bochner integral} \label{B-integral}
%%%%%%%%%%%%%%%%%%%%%%%%%%%%%%%%%%%%%%%%%%%%%%%%%%%%
The Bochner integral can generally define a vector-valued or operator-valued integral in a complete normed space (Banach space). Here, we define a ${\cal H}$-valued integral and prove two basic theorems. %note a few important properties. 

Let $D$ be a compact domain in $\mathbb C$. 
We say a ${\cal H}$-valued function $\{\ket {f_\alpha}\}_{\alpha \in \mathbb C}$ is \textit{Bochner integrable} on $D$ if
there exists a sequence of ${\cal H}$-valued simple functions $(s_n)$  that satisfies %such that
\begin{align}
%\text{(i)} &\quad \lim_{n \to \infty } \| s_n(\alpha) -\ket{f_\alpha} \| =0 \quad \alpha \in D \text{ a.e.},  %\quad (n \to \infty ) 
%\end{align} holds for $\alpha \in D$ a.e., % \text{ a.e.}$
%and  it holds
%\begin{align}
%&  \| s_n(\alpha) -\ket{f_\alpha} \| \to  0 \quad (n \to \infty ),  \  \alpha \in D \text{ a.e.},   \\ 
%\\
%\text{(ii)} &\quad  
\lim_{n \to \infty } \int_D  \| s_n(\alpha) -\ket{f_\alpha} \|^2d^2\alpha =0 \label{BintegrabilityAA} %\to 0 \quad (n \to \infty ) %\tag{*X}
. \end{align} 
In such a case, the integral of $\{\ket {f_\alpha}\}_{\alpha \in \mathbb C}$ is defined as 
\begin{align}
 \int_D \ket {f_\alpha} d^2 \alpha :=\lim_{n \to\infty} \int_D s_n(\alpha) d^2\alpha. 
\end{align}
We may write
\begin{align}
 \int_{\mathbb{C}} \ket {f_\alpha} d^2 \alpha :=\lim_{|D| \to\infty} \int_D \ket {f_\alpha} d^2\alpha ,
\end{align} whenever the right hand side exists. 

\begin{remark}
 %Here, a step function means its range is a finite set. %a function which takes finite valuefinite linear combination of characteristic functions on 2-dimensional intervals. 
The relation in Eq.~\ref{BintegrabilityAA} implies  
\begin{align}
%\text{(i)} &\quad 
\lim_{n \to \infty } \| s_n(\alpha) -\ket{f_\alpha} \| =0 \quad \text{ a.e. } \alpha \in D .
\end{align}
\end{remark}

%\begin{remark}
% Here, a step function means a finite linear combination of characteristic functions on 2-dimensional intervals. 
%\end{remark}
We would like to show the integral defined in such a way actually belongs to ${\cal H}$. It turns out the Bochner integrability of ${\cal H}$-valued functions coincides the square integrability of the norm of those functions. %r implies the square integrability of the associated norm 
Similarly to the case of $L^2$ functions, the following theorem holds (See, e.g., Theorem V.5.1 in Ref.~\cite{FAYoshida}).

\begin{theorem}\label{NFBochner}
An ${\cal H}$-valued function $\ket {f_\alpha}$ is Bochner integrable on $D$  iff \begin{align}
 \int_D \| \ket {f_\alpha} \| d^2 \alpha < \infty.  
\end{align}
% $\int_D \| \ket {f_\alpha} \| d^2 \alpha < \infty$
\end{theorem}
\begin{proof}
%See, e.g., Theorem V.5.1 in Ref.~\cite{FAYoshida}.% Functional analysis  K. Yoshida 
Let  us assume assume $\int_D \| \ket {f_\alpha} \|^2 d^2 \alpha < \infty$. We have skipped details but the functions should be assumed measurable. So, we  can take a sequence of simple functions $(s_n)$ that satisfies 
\begin{align} \lim _{n \to \infty }\| \ket{ f_ \alpha} - s_n (\alpha ) \| = 0 , \quad a.e.,  \ \alpha \in D.
 \end{align}
%$\|f_ \alpha - s_n (\alpha ) \| = 0 , \quad a.e.,  \ \alpha \in D$
 In what follows we may drop to write the condition ``$(a.e. ),\  \dots $ '' as it is irrelevant to the main points of the proof. %From the construction it holds $\lim _{n \to \infty } \int \| s_n\|^2 d^2 \alpha =\int \|  \ket{ f_ \alpha} \|^2 d^2 \alpha $,  and thus $(\int \|s_n\|^2 d^2 \alpha)$  defines a convergent sequence. 

Let us define another sequence of simple functions as 
\begin{align}t_n (\alpha ) =  \begin{cases}  s_n (\alpha )& \|s_n (\alpha) \| \le 2\| \ket{f_\alpha }  \| \\  0 & \text(otherwise) \end{cases} .
\end{align}
This sequence satisfies %$\lim_{n \to \infty } \| t_n - s_n \| = 0$ and  $ \lim_{n \to \infty } \| t_n - f \| = 0 $
\begin{align}
 \lim_{n \to \infty } \| t_n (\alpha )- s_n (\alpha ) \| = 0   \quad a.e. \ \alpha \in D, \\  
  \lim_{n \to \infty } \| t_n (\alpha ) -  \ket{ f_ \alpha} \| = 0  \quad a.e. \ \alpha \in D.  \label{aeZerotf}
 \end{align}
From  the triangle inequality and  the construction of $(t_n)$, it holds that   $\|f  - t_n  \| \le \|f \| + \|t_n \| \le 3 \|f \| $. We thus have 
\begin{align}
\| \ket{ f_ \alpha}   - t_n (\alpha ) \| ^2 \le 9 \| \ket{ f_ \alpha}  \|^2.   
\end{align} This implies $\| f -t_n\|$ is square-integrable because $ \| f\|^2 $  is square-integrable. We thus can use the dominated convergence theorem to obtain  
\begin{align}
&  \lim_ { n \to\infty }  \int \|\ket{f_\alpha }  - t_n (\alpha ) \| ^2d^2 \alpha   \nonumber \\ = &      \int  \lim_ { n \to \infty } \|\ket{f_\alpha }  - t_n (\alpha ) \| ^2d^2 \alpha   =  0, 
\end{align}
where the final equality is  due to the relation in Eq.~\ref{aeZerotf}. %  for the final equality. 
Since $(t_n)$ plays the role of $(s_n)$ in \ref{BintegrabilityAA}, the "if" part of the theorem statement is proven.

%\begin{align}\end{align}
Let us move on to the ``only if'' part. 
From the triangle inequality $\|f \| \le \|f -s_n \|  +\|s_n \| $, we have 
\begin{align}
\|f \|^2 \le & \|f -s_n \|^2  +\|s_n \| ^2 + 2 \| f-s_n \|\| s_n \|  \nonumber \\
\le & 2 \|f -s_n \|^2  +2 \|s_n \| ^2 .  \end{align}
This implies 
\begin{align}
\int \|f \|^2  d^2 \alpha 
\le & 2 \int  \|f -s_n \|^2 d^2 \alpha  +2  \int \|s_n \| ^2 d^2 \alpha . \end{align}
Due to the condition in Eq.~\ref{BintegrabilityAA}, the term $\int \|f -s_n \|^2 d^2 \alpha   $ is a convergent  sequence and bounded. Boundedness of the last term  $\int \|s_n \| ^2 d^2 \alpha$ is guaranteed because it is an integral of a simple function defined over the compact domain $ D$. This proves
$\int \|f \|^2  d^2 \alpha < \infty  $.
\end{proof}

We may often use the triangle inequality for the vector-valued integrals.
\begin{theorem} \label{sankakuIntTheo} It holds %Triangle inequality for integrals)
\begin{align}\left \| \int  \ket { f_ \alpha  }   d^2 \alpha  \right \| \le \int \|  \ket {f_\alpha} \| d^2 \alpha  ,
\end{align} whenever both sides exist. %side exists.
\end{theorem}

\begin{proof}%\noindent \textit{Proof.}
If $ \int \ket{f_\alpha } d^2\alpha = 0$ it is obviously true. Let us suppose $\int \ket{f_\alpha } d^2 \alpha  \neq 0$.
 Let us define a unit vector 
 as 
\begin{align}
e_f := \frac{\int \ket{f_\alpha } d^2 \alpha }{\| \int  \ket{f_\alpha } d^2 \alpha  \| }.%  \lim_{n \to \infty }   \int ( f -t_n )    = \int f -  \lim _{n \to \infty} \int   t_n    =  0,
\end{align}  %and equivalently, 
We can see that 
\begin{align}
\left \| \int  \ket{f_\alpha } d^2 \alpha   \right \| % \braket {e_f | \int  \ket{f_\alpha } d^2 \alpha  } 
= & \bra {e_f} \cdot  \left (  \int \ket{ f_\alpha  } d^2 \alpha  \right )   \nonumber \\
= & \int \braket{e_f  | f_\alpha  } d^2 \alpha    \nonumber \\
\le & \int \sup _{\| \varphi \| \le 1 } \braket{\varphi | f_\alpha  } d^2 \alpha %\nonumber \\
        =  \int \|\ket{f _\alpha }\| d^2 \alpha %  \lim_{n \to \infty }   \int ( f -t_n )    = \int f -  \lim _{n \to \infty} \int   t_n    =  0,
\end{align}
\end{proof}%\hfill$\blacksquare$  

\begin{remark}
The Bochner integrability of $\ket{f_\alpha}$ does not necessary guarantee the integrability of $\| \ket {f_\alpha}\|$. 
The square integrability of $\| \ket {f_\alpha}\|$ implies its integrability when the  integration volume is finite. In fact, Schwartz inequality helps us to obtain  
\begin{align}%\left \| \int  \ket { f_ \alpha  }   d^2 \alpha  \right \| \le \
\int_D \|  \ket {f_\alpha} \| d^2 \alpha  \le & \nonumber \sqrt {\int_D d^2 \alpha } \sqrt{\int_D \left \| \ket { f_ \alpha  }  \right \|^2  d^2 \alpha | } \\
= &  {|D|}^{1/2} \sqrt{\int _D\left \| \ket { f_ \alpha  }  \right \|^2  d^2 \alpha |}  < \infty ,
\end{align} where $D$ is assumed to be a compact region on $\mathbb C$.
\end{remark}

%%%%%%%%%%%%%%%%%%%%%%%%%%%
\section{elementary integration} \label{AppIntegral1}
%%%%%%%%%%%%%%%%%%%%%%%%%%%

Let us define 
\begin{align}
 I_{n}( R) :=\int ^{R}_{0}\dfrac{y^{n}e^{-y}}{n!}dy, \qquad ( n=0,1,2,\dots ).
\end{align}
As the integrand is positive, it holds $I_{n}\left( R\right) \geq 0 $ for $R > 0$. 
Integration by parts yeilds  
\begin{align}
I_{n}\left( R\right) =I_{n-1}\left( R\right) -\dfrac{R^{n}e^{-R}}{n!} \label{tag2}.
\end{align}
This relation leads to 
\begin{align}
I_{n}\left( R\right) \le I_{n-1}\left( R\right) \label{tag3}. 
\end{align}
We can readily confirm
\begin{align}
I_{0}\left( R\right) = \int_0^R e^{-y}dy = 1-e^{-R}\le 1 \label{tag4}.
\end{align}
Let $\chi_{[0,R]}$ be the characteristic function on the interval $[0,R]$. 
Applying the  monotone convergence theorem to the sequence of functions $f_m(y) =e^{-y } \chi_{[0,m]} (y) $, we have  
\begin{align}
 \int_{\mathbb R} \lim_{m \to\infty}  f_{m}  (y)dy  =  \lim_{R\to \infty}I_{0}\left( R\right) = \lim_{R\to \infty}\left(  1- e^{-R} \right )  =1. \label{Qn221025A5}
\end{align}
%\begin{align}
%\lim_{R\to \infty}I_{0}\left( R\right) = \lim_{m \to\infty} \int_{\mathbb R} f_{m}  (y)dy  = \int_0^\infty e^{-y}dy =1. \label{Qn221025A5}
%\end{align}

Now, let us define 
%the sequences of functions as
\begin{align}
 f_m^{(n)}(y) =(n!)^{-1}  y^n e^{-y } \chi_{[0,m]} (y). 
\end{align} 
For $n=1$, from Eq.~\eqref{tag2} and   %,~\eqref{tag3} and \eqref{Qn221025A5}, we can use 
the monotone convergence theorem, we obtain 
\begin{align}
 \int_{\mathbb R} \lim_{m \to\infty}  f_{m}^{(1)} (y)dy =&   \lim _{R \to \infty}I_{1}  \left( R\right) \nonumber \\
  = & \lim _{R \to \infty} \left ( I_{0}  \left( R\right) - R e^{-R} \right)   =1. 
\end{align}
%\begin{align}
% \lim _{R \to \infty}I_{1}  \left( R\right)= \lim_{m \to\infty} \int_{\mathbb R} f_{m}^{(1)} (y)dy  =1. 
%\end{align}
Repeating this process for $n =2,3,4, \dots$, we obtain 
\begin{align}
 \int_{\mathbb R} \lim_{m \to\infty}  f_{m}^{(n)} (y)dy =&   \lim _{R \to \infty}I_{n}  \left( R\right) \nonumber \\
  = & \lim _{R \to \infty} \left ( I_{n-1}  \left( R\right) - \frac{ R^n e^{-R}}{n!} \right)   =1. 
\end{align}
%\begin{align}
% \lim _{R \to \infty}I_{n}\left( R\right) =1, \qquad (n=2,3,4, \dots ). 
%\end{align}
Note that Eqs.~\eqref{tag3}~and~\eqref{tag4} readily imply  
\begin{align}
 |1- I_n (R)|<1 \label{tag5}, \quad ( R \ge 0). %R \in (0, \infty). 
\end{align}

%%%%%%%%%%%%%%%%%%%%
\section{detail of calculation} \label{somewhatlengthy}
%%%%%%%%%%%%%%%%%%%%
Here, we show the following relation:
\begin{align} 
\int_{D(R)} \left| \alpha \rangle \langle \alpha \right| \varphi \rangle d^{2}\alpha =    \pi \sum ^{\infty }_{n=0}I_n(R^2) \varphi_n \ket{n}. \label{Qn22B1}% \\
 %=\sum ^{\infty }_{n=0}\left(  \int_{D(R)} e^{-|\alpha |^2/2} \dfrac{\alpha ^{n}}{\sqrt{n!}}\langle \alpha | \varphi \rangle d^{2}\alpha \right) | n\rangle 
 \end{align}
 
 Let be $D(r) =\set{\alpha\in \mathbb C | \left |\alpha\right|\le r}$. The number state expansion of $\ket {\alpha}$ in Eq.~\eqref{coState} implies 
\begin{align} 
 \int_{D(R)} \left| \alpha \rangle \langle \alpha \right| \varphi \rangle d^{2}\alpha =&\int_{D(R)}  e^{-|\alpha |^2/2} \sum ^{\infty }_{n=0}\dfrac{\alpha^n }{\sqrt{n!}} | n \rangle \langle \alpha | \varphi \rangle d^{2}\alpha \nonumber \\
 = & \sum ^{\infty }_{n=0}\left(  \int_{D(R)} e^{-|\alpha |^2/2} \dfrac{\alpha ^{n}}{\sqrt{n!}}\langle \alpha | \varphi \rangle d^{2}\alpha \right) | n\rangle, \label{Qn22B2}
 \end{align}
where in the last line we use Theorem~\ref{Th4G} in Appendix~\ref{APPth4G} to exchange the order of integration and summation for  $\cal H $-valued terms  (Note that the assumptions of Theorem~\ref{Th4G} are fulfilled as $|\braket{\alpha|\varphi}|$ is uniformly bounded). % on $\alpha \in D(R)$). % {APPth4}. %In what follows, the exchange of integration and the summation over the index $m$ is due to the dominated convergence theorem, where $R >0$ is finite.  volum Then the integration with 

Now, let us consider the following integration: % with another number-state expansion on $\braket{\alpha|\varphi}$: 
\begin{align}
\nonumber  & \int_{D(R)} e^{-|\alpha |^2/2} \dfrac{\alpha ^{n}}{\sqrt{n!}}\langle \alpha | \varphi \rangle d^{2}\alpha \\ 
= & \int _{|\alpha| \le R} \left( e^{-\left| \alpha \right| ^{2}}\dfrac{\alpha ^{n}}{\sqrt{n!}}\sum ^{\infty }_{m=0}\dfrac{\left( \alpha ^{\ast }\right) ^{m}\varphi _{m}}{\sqrt{m!}}\right) d^{2}\alpha  .\label{Qn22B3}
\end{align}
Using Schwartz's inequality, we can show the power series is uniformly bounded as
\begin{align} 
\left | \sum ^{ N}_{m=0} \dfrac{\left( \alpha ^{\ast }\right) ^{m}\varphi _{m}}{\sqrt{m!}}  \right | \le & \left( \sum_{m=0}^{N}\frac{|\alpha |^{2m}}{m!} \right)^{1/2} \left( \sum_{m=0}^{N} |\varphi_m |^2 \right)^{1/2}  \nonumber \\
\le &  e^{|\alpha |^2/2} \|\varphi \|  \le    e^{R^2/2} \|\varphi \| .
\end{align}
Hence, the integrand is a uniform limit of a sequence of continuous functions. %which is dominated by a convergent sequence
%\begin{align} \sum  \left | e^{-|\alpha |^2/2} \dfrac{\alpha ^{n}}{\sqrt{n!}}\langle \alpha   |  \varphi \rangle \right | \le \dfrac{R^{n}}{\sqrt{n!}}   \|\varphi \|,   \quad | \alpha |\le R, \end{align} 
This allows us to %we can use the bounded convergence theorem so as to 
exchange the order of the integration and the summation in the second expression of Eq. \ref{Qn22B3}. We thus obtain    
\begin{align}
 & \int _{|\alpha| \le R} \left( e^{-\left| \alpha \right| ^{2}}\dfrac{\alpha ^{n}}{\sqrt{n!}}\sum ^{\infty }_{m=0}\dfrac{\left( \alpha ^{\ast }\right) ^{m}\varphi _{m}}{\sqrt{m!}}\right) d^{2}\alpha   \nonumber \\
% = & \int _{|\alpha| \le R} \left( e^{-\left| \alpha \right| ^{2}} \sum ^{\infty }_{m=0}\dfrac{\alpha ^{n} \left(  \alpha ^{\ast }\right) ^{m}\varphi _{m}}{\sqrt{n! m!}}\right) d^{2}\alpha  | n\rangle \\
% = &  \int _{|\alpha| \le R} \sum ^{\infty }_{m=0} \left( e^{-\left| \alpha \right| ^{2}}
  %\dfrac{\alpha ^{n} \left(  \alpha ^{\ast }\right) ^{m}\varphi _{m}}{\sqrt{n! m!}}\right) d^{2}\alpha  \nonumber \\
 = & \sum ^{\infty }_{m=0}\left(  \dfrac{\varphi _{m}}{\sqrt{n! m!}} \int _{|\alpha| \le R}   e^{-\left| \alpha \right| ^{2}} \alpha ^{n}   \left(\alpha ^{\ast }\right) ^{m}
  d^{2}\alpha \right) \nonumber \\
 =&  \sum ^{\infty }_{m=0}\left(  \dfrac{\varphi _{m}}{\sqrt{n! m!}} \int _0  ^ R   e^{-r ^{2}}  r^{n+m} rdr  \cdot \underbrace{\int_0^{2 \pi} e^{i(n-m) \phi} d\phi}_{2 \pi \delta_{n,m}}
 \right)  \nonumber \\
 =&    \sum ^{\infty }_{m=0}\left(  \dfrac{\varphi _{m}}{\sqrt{n! m!}} \int _0  ^ R   e^{-r ^{2}}  r^{n+m} rdr  \ 2 \pi \delta_{n,m}
 \right) \nonumber \\
=&   \pi  \dfrac{\varphi _{n}}{n!} \int _0  ^ {R^2}   e^{-y}  y^n dy 
 =  \pi  I_n(R^2) 
   \varphi _{n}, \label{Qn22B9}
\end{align}
where  %we exchange the sum and the integral in the third line, and  the rest part follows from 
an integration in the polar coordinate system was carried out with $\alpha = r e^{i\phi}$. % leads to the final expression. 
 Concatenating  Eqs.~\ref{Qn22B2}, \ref{Qn22B3}, and \ref{Qn22B9}, we find the relation in Eq.~\ref{Qn22B1}. % holds.  

%%%%%%%%%%%%%%%%%%%%%%%%%%%%%%%%%%%%%%%%%%%%%%%%%%%%
\section{action of bra} \label{BraAction}
%%%%%%%%%%%%%%%%%%%%%%%%%%
One may not quite sure if the action of bra can be put into the ${\cal H}$-valued integral
 as
 \begin{align}
 \bra{f} \int  \ket{g_\alpha}  d^2\alpha =  \int  \braket{f|g_\alpha}  d^2\alpha. 
\end{align}
We can show that the following theorem holds: 
\begin{theorem}\label{Theo10000}
%Let be \begin{align} \ket{\psi_\alpha }:= \ket{\alpha}\braket{\alpha|\varphi} \end{align}
Let be $ f , \psi_\alpha \in {\cal H} \ (\alpha \in \mathbb C)$ and let us assume 
 \begin{align}
 \int_{\mathbb C}\ket{\psi_\alpha }d^2 \alpha \in {\cal H}, \qquad  \int_{\mathbb C}\braket{f|\psi_\alpha }d^2 \alpha \in \mathbb C .  \label{Qn2210240}
\end{align}
Then, it holds 
\begin{align}
 \bra f \int_{\mathbb C}\ket{\psi_\alpha }d^2 \alpha = \int_{\mathbb C}\braket{f|\psi_\alpha }d^2 \alpha. 
\end{align}
\end{theorem}

Here, we prove a corollary of Theorem~\ref{Theo10000}, in which  $\ket{\psi_\alpha}$ is norm-continuous with respect to $\alpha$, namely, $\|\ket{\psi_\alpha} -\ket{\psi_\beta }  \| \to 0  \quad (\|\alpha -\beta \| \to 0)$,  so that the integration can be well-approximated with Riemann sums. To prove Theorem~\ref{Theo10000}, one can remove the continuity assumption by replacing the Riemann-sum argument with an argument based on simple functions. 
\begin{corollary}
Let $\ket{\psi_\alpha}$ be norm-continuous with respect to $\alpha$, namely, $\|\ket{\psi_\alpha} -\ket{\psi_\beta }  \| \to 0  \quad (\|\alpha -\beta \| \to 0)$. Let us assume 
 \begin{align}
 \int_{\mathbb C}\ket{\psi_\alpha }d^2 \alpha \in {\cal H}, \qquad  \int_{\mathbb C}\braket{f|\psi_\alpha }d^2 \alpha \in \mathbb C .  \label{Qn221024}
\end{align}
  Then, it holds 
\begin{align}
 \bra f \int_{\mathbb C}\ket{\psi_\alpha }d^2 \alpha = \int_{\mathbb C}\braket{f|\psi_\alpha }d^2 \alpha. 
\end{align}
\end{corollary}

\begin{proof} Let be $D(r) =\set{\alpha\in \mathbb C | \left |\alpha\right|\le r}$. For a convention to write a Riemann sum, let  $\Delta = (\Delta_m)_{m=1}^M$ denote the partition of $D(r)$ such that 
\begin{align}
D(r) =\bigcup_{m=1 }^M\Delta_m ,\quad\mu(\Delta_i\cap\Delta_j)=0 \quad (i \neq j),  \label{Qn2210242}\\ |\Delta|:=\max_m \sup_{\alpha, \beta \in \Delta_m} |\alpha -\beta|.  \label{Qn2210243}
\end{align}

From the assumption in Eq.~\eqref{Qn221024}, for any $\epsilon >0 $ there exists $R > 0$ such that, for  $r \ge  R$, it holds 
\begin{align}
&  \left | \bra f \left ( \int_{\mathbb C}\ket{\psi_\alpha }d^2 \alpha - \int_{D(r)}\ket{\psi_\alpha }d^2 \alpha \right) \right | \nonumber \\
\le&  \|f \|  \quad \left \| \int_{\mathbb C}\ket{\psi_\alpha }d^2 \alpha - \int_{D(r)}\ket{\psi_\alpha }d^2 \alpha \right\| < \epsilon /4,   \end{align}
and   
 \begin{align}
& \left |   \int_{\mathbb C}\braket{f|\psi_\alpha }d^2 \alpha - \int_{D(r)}\braket{f|\psi_\alpha }d^2 \alpha  \right | < \epsilon /4.
\end{align}

Moreover, from the conditions 
\begin{align}
 \int_{D(r)}\ket{\psi_\alpha }d^2 \alpha \in {\cal H}, \qquad  \int_{D(r)}\braket{f|\psi_\alpha }d^2 \alpha \in \mathbb C, 
\end{align}
 there exists $\delta>0$ such that,  for any partition $ (\Delta_m)_{m=1}^M$ of $D(r)$ satisfying  % $D(r) = \bigcup_m \Delta_m$ %
 Eq.~\eqref{Qn2210242} and $|\Delta | <\delta$, it holds 
\begin{align}
&  \left | \bra f \left ( \int_{D(r)}\ket{\psi_\alpha }d^2 \alpha  - \sum_m \ket{\psi_{\alpha_m}} \mu(\Delta_m) \right) \right | \nonumber \\
\le&  \|f \|  \quad \left \| \int_{D(r)}\ket{\psi_\alpha }d^2 \alpha  - \sum_m \ket{\psi_{\alpha_m}} \mu(\Delta_m) \right \| < \epsilon /4  
\end{align} and 
\begin{align}
 & \left |   \int_{D(r)}\braket{f|\psi_\alpha }d^2 \alpha   - \sum_m \braket{f|\psi_{\alpha_m}} \mu(\Delta_m)  \right | < \epsilon /4
\end{align} where $\alpha_m \in \Delta_m$. 

Now, noting  that for any finite sum, it is no problem to write 
\begin{align} 
  \bra f \left ( \sum_m \ket{\psi_{\alpha_m}} \mu(\Delta_m) \right) = \sum_m \braket{f|\psi_{\alpha_m}} \mu(\Delta_m)  , 
\end{align}
we can make a chain of triangle inequalities to show
\begin{align} 
&\left | \bra{f}\int_{\mathbb C} \ket{ \psi_\alpha} d^{2}\alpha - \int_{\mathbb C }\braket{f| \psi_\alpha }  d^{2}\alpha \right | \nonumber \\
 \le& \left | \bra{f}\int_{\mathbb C} \ket{ \psi_\alpha} d^{2}\alpha 
-\bra{f}\int_{D(r)} \ket{ \psi_\alpha} d^{2}\alpha \right | \nonumber \\
&+
\left | \bra f \left ( \int_{D(r)}\ket{\psi_\alpha }d^2 \alpha  - \sum_m \ket{\psi_{\alpha_m}} \mu(\Delta_m) \right) \right |
\nonumber \\
&+
\left | \bra f \left ( \sum_m \ket{\psi_{\alpha_m}} \mu(\Delta_m) \right) -\sum_m \braket{f|\psi_{\alpha_m}} \mu(\Delta_m)  \right |  \nonumber \\ 
&+\left |    \sum_m \braket{f|\psi_{\alpha_m}} \mu(\Delta_m) - \int_{D(r)}\braket{f|\psi_\alpha }d^2 \alpha  \right | \nonumber \\
&+
\left| \int_{D(r) }\braket{f| \psi_\alpha }  d^{2}\alpha  
- \int_{\mathbb C }\braket{f| \psi_\alpha }  d^{2}\alpha\right  | \nonumber \\
<&\epsilon/4 + \epsilon/4 + 0 +\epsilon/4 + \epsilon/4 = \epsilon \end{align}

\end{proof}

%%%%%%%%%%%%%%%%%%%%%%%%%%%%%%%%%%%%%%%%%%%%%%%%%%%%%%%%%%%%%%%%%%%%%%%%%%%%
\section{order of integral and summation}\label{APPth4G}
%%%%%%%%%%%%%%%%%%%%%%%%%%%%%%%%%%%%%%%%%%%%%%%%%%%%%%%%%%%%%%%%%%%%%%%%%%%%

\begin{theorem} \label{Th4G}
Let  $ {\cal H} = \ell ^2$ and $(\ket{n})_{n=0}^\infty $ be an orthonormal basis  on ${\cal H}$. Let be $D(r) =\set{\alpha\in \mathbb C | \left |\alpha\right|\le r}$ and $ (\varphi _{n})_{n =0}^ \infty $ is a squence of complex-valued functions on $D$ which fulfills  \\ 
%(i) $\sum ^{\infty }_{n=0} \varphi _{n} ( \alpha ) | n\rangle  \in {\cal H}$ for each $\alpha \in \mathbb C$ and 
\begin{align}
\text{(i) } && \sum ^{\infty }_{n=0} | \varphi _{n} ( \alpha ) |^2< \infty,  %\sum ^{\infty }_{n=0} \varphi _{n} ( \alpha ) | n\rangle  \in {\cal H} 
\quad  a.e., \ \alpha \in D(r)  %\subset \mathbb C
 \\ 
\text{(ii)} &&  \int_{D(r)} \sum ^{\infty }_{n=0}  \left | \varphi _{n} ( \alpha ) \right |^2 d^2 \alpha < \infty % \ket{\phi}:=
%\int_{D(r)} \left( \sum ^{\infty }_{n=0}  \varphi _{n}\left( \alpha \right) | n\rangle  \right ) d^{2}\alpha   \in {\cal H}
\label{B1kana}
.  \end{align}% with some complex-valued function and assume $\phi \in {\cal H}$. 
Then, it holds that \begin{align}
 \int_{D(r)} \left(  \sum ^{\infty }_{n=0}\varphi _{n}\left( \alpha \right) | n\rangle \right)  d^{2}\alpha  =\sum_{n=0} ^{\infty} \left(\int_{D(r)} \varphi _{n} \left( \alpha \right)  d^{2}\alpha \right) | n\rangle. \label{B22kana}
\end{align}
\end{theorem}
%\begin{remark}
%This can be immediately proven via the dominated convergence theorem for Bochner integral. As a notation onvention in the following proof, we may write
%\begin{align}
%\|\varphi\|_{L^2[D]}^2 := \int_D |\varphi(\alpha)|^2 d^2 \alpha.  
%\end{align}
%\end{remark}

\begin{remark}
In order to verify the expression in Eq.~\ref{Qn22B2} within an elementary framework, one can use Theorem~\ref{Th4} instead of Theorem~\ref{Th4G}. Theorem~\ref{Th4} represents the case in which the sequence of the functions $(\varphi_n)$ is associated with a power series, i.e., $\varphi_n (\alpha ) =a_n \alpha^n $ . 
In such a case, % the continuity of the analytic function is at least contiuous 
the vector-valued integral can be defined as a limit of a Riemann sum similarly to Remark~\ref{remark20}. In this manner, one can work out our main theorem (Theorem \ref{Theorem10}) without concerning the notion of the Bochner integral as well as  that of  the Lebesgue integral.
\end{remark}

%\begin{remark}
%In order to verify the expression in Eq.~\eqref{Qn22B2} within an elementary framework, let us observe that the integrand in Eq.~\eqref{Qn22B2} is at most an analytic function of $\alpha $, and let us recall that term by term integration is valid for an absolutely integrable power series. Then, we can modify the following proof by defining $\varphi_n (\alpha ) :=a_n \alpha^n $ so as to prove the statement of our theory when the functions are associated with a power series (See Appendix~\ref{APPth4}). Thereby, % the continuity of the analytic function is at least contiuous 
%the vector-valued integral can be defined as a limit of Riemann sum similarly to Remark~\ref{remark20}. In this manner, one can work out our main theorem (Theorem \ref{Theorem10}) without concerning the notion of Bochner integral as well as lebesgue integral.
% % In other words, Essentially, an analytic function is continuous, we Integration term by term together with the completeness of the hilbert space guarantee the theorem in the case of a power series.
%\end{remark}

%\newpage

\begin{proof}
By the monotone convergence theorem for  $\left ( \sum_{n=0}^k |\varphi_n|^2 \right ) \in L^1[D] $
and the condition  (ii), we have  
\begin{align}
 \int_{D}   \sum _{n = 0} ^\infty \left| \varphi _{n}\left( \alpha \right)\right |^2  d^{2}\alpha  = & \sum  _{n= 0}^\infty \left(\int_{D } \left| \varphi _{n}\left( \alpha \right)\right |^2  d^{2}\alpha \right) \nonumber \\ 
 =&   \sum_{n=0}^\infty \| \varphi_n \|_{L^2[D]}^2 < \infty. \label{T4inQ1} 
\end{align} 
This implies 
\begin{align}
 \int \left \| \sum \varphi_n (\alpha ) \ket n \right \|^2  d^2 \alpha = & \int { \sum\left | \varphi_n (\alpha )\right |^2 }  d^2 \alpha % \nonumber \\ \le &  |D|^{1/2} \sqrt{ \int  \sum\left | \varphi_n (\alpha )\right |^2 d^2 \alpha } \nonumber \\   = &  |D|^{1/2} \sqrt{   \sum \left \| \varphi_n \right \|_{L^2[D]}^2  } 
  < \infty. \label{QN20230123a}
\end{align}
 Since the Bochner integrability is fulfilled due to Theorem~\ref{NFBochner},  the integrated state vector in the following form exists,
\begin{align}
 \phi :=& \int_{D}   \sum _{n = 0} ^\infty   \varphi _{n}\left( \alpha \right) \ket{n} \ d^{2} \alpha. 
\end{align}
Thereby, integrals of truncated states in the following form exist, 
\begin{align}
 \phi^{(N)} :=& \int_{D}   \sum _{n = 0} ^N  \varphi _{n}\left( \alpha \right) \ket{n} \ d^{2} \alpha. 
 \end{align}
Since the series $\sum \left \| \varphi_n \right \|_{L^2[D]}^2  $ is convergent as shown in Eq.~\eqref{T4inQ1}, we can show  %$ \phi^{(N)} \to \phi $ as  %{QN20230123a}. 
\begin{align}
 \| \phi - \phi^{(N-1)}  \|   =& \left\| \int_{D}   \sum _{n = N} ^\infty   \varphi _{n}\left( \alpha \right) \ket{n} \ d^{2} \alpha  \right\| \nonumber \\
  \le& \int_{D}  \left \|   \sum _{n = N} ^\infty   \varphi _{n}\left( \alpha \right) \ket{n} \right\| d^{2} \alpha \nonumber \\
\le &   |D|^{1/2} \sqrt{   \sum_{n=N}^\infty  \left \| \varphi_n \right \|_{L^2[D]}^2  } \to 0  \quad (N \to \infty). \label{Qn20230123b}
\end{align}  where we use the triangle inequality (Theorem~\ref{sankakuIntTheo}) and Schwartz's inequality.

Now, let us define 
\begin{align}
 \mathbb C \ni a_n :=  \left(\int_{D }  \varphi _{n}\left( \alpha \right)  d^{2}\alpha \right). 
\end{align}
The sequence $(a_n)$ is square-summable as one can show  
\begin{align}
 |a_n| =&  \left|\int_{D }  \varphi _{n}\left( \alpha \right)  d^{2}\alpha \right| 
\le\int_{D } \left| \varphi _{n}\left( \alpha \right) \right|  d^{2}\alpha \nonumber  \\
 \le&  \sqrt{\int_{D } d^{2}\alpha}  \sqrt{\int_{D } \left| \varphi _{n}\left( \alpha \right) \right|^2  d^{2}\alpha} \le |D|^{1/2}  \ \| \varphi_n\|_{L^2 (D)} \nonumber % < \infty
\end{align} and 
\begin{align}
 \sum_n|a_n|^2  \le & |D|\sum_n   \| \varphi_n\|_{L^2 (D)}^2 < \infty  %\nonumber \\ = &|D|^2 \sum_{n}\int_{D } \left| \varphi _{n}\left( \alpha \right) \right|^2  d^{2}\alpha < \infty,  
\end{align} where we use Eq.~\eqref{T4inQ1} in the final inequality. 
Therefore, the state vector in the form of  
\begin{align}
 \psi^{(N)} =\sum_{n=0}^N a_n \ket{n} %\in {\cal H}
\end{align} defines a Cauchy sequence $(\psi ^{(N)})$ in ${\cal H }$, and its   unique limit $\psi \in {\cal H}$ is well-defined due to the completeness of ${\cal H}$. It thus holds  
\begin{align}
 \|\psi - \psi^{(N)}\| \to 0 \qquad (N \to \infty). \label{Qn20230123c}
\end{align}
Since $\| \phi^{(N)} - \psi^{(N)} \| =0$ for $N \in \mathbb N$,  Eqs.~\eqref{Qn20230123b}~and~\eqref{Qn20230123c} yield % ,  In turn, $\phi^{(N)}$ and  $ \psi^{(N)}  $ define the same sequence as thier summation is  finite, and converge to the same limit. We thus obtain 
\begin{align}
 \left \| \phi - \psi \right\| = 0 , 
\end{align} namely, it holds 
\begin{align}
 \int_{D}   \left(  \sum _{n = 0} ^\infty   \varphi _{n}\left( \alpha \right) \ket{n} \right) d^{2} \alpha  = \sum  _{n= 0}^\infty \left(\int_{D }  \varphi _{n}\left( \alpha \right)  d^{2}\alpha \right) \ket{n}. 
\end{align}
%  $\| \phi^{(N)} - \psi^{(N)} \| =0$ implies $( \phi^{(N)})$ is also cauchy, and its limit $\phi $ has to satisfy
%\begin{align}
% \|\phi - \phi^{(N)}\| \to 0 \qquad (N \to \infty). 
%\end{align}
\end{proof}

\section{An Elementary Approach to Exchange the Order of Integral and Summation for $\cal H$-valued Integrals Associated with Power Series}\label{APPth4}
%%%%%%%%%%%%%%%%%%%
Here we will show a type of the dominated convergence theorem for $\cal H$-valued integrals when an associated sequence of functions is given by a power series. This theorem is proven to verify the relation in Eq.~\ref{Qn22B2}  without invoking neither the Bochner integral nor the Lebesgue integral. Its generalized version is  Theorem~\ref{Th4G}, whose proof necessitates  the Bochner integrability and the monotone convergence theorem. % given in Appendix~\ref{APPth4G}.  % is shown.  
\begin{theorem} \label{Th4}
Let  $ {\cal H} = \ell ^2$ and $(\ket{n})_{n=0}^\infty $ be an orthonormal basis  on ${\cal H}$. Let be $D(r) :=\set{\alpha\in \mathbb C | \left |\alpha\right|\le r}$ and $  \sum_{k=0}^n a_{k} \alpha ^k $ %  with $(a_n)_{n=0}^\infty \in \mathbb C$ 
be a power series % on $D$
 which fulfills  \\ 
%(i) $\sum ^{\infty }_{n=0} \varphi _{n} ( \alpha ) | n\rangle  \in {\cal H}$ for each $\alpha \in \mathbb C$ and 
\begin{align}
\text{(i) } && \sum ^{\infty }_{n=0} \left | a _{n} \right |^2  \left  |\alpha \right  |^{2n}< \infty,  %\sum ^{\infty }_{n=0} \varphi _{n} ( \alpha ) | n\rangle  \in {\cal H} 
\quad \alpha \in D(r) ,  %\subset \mathbb C
 \\ 
\text{(ii)} &&  \int_{D(r)} \sum ^{\infty }_{n=0}  \left | a_{n} \right |^2 \left| \alpha \right |^{2n} d^2 \alpha < \infty % \ket{\phi}:=
%\int_{D(r)} \left( \sum ^{\infty }_{n=0}  \varphi _{n}\left( \alpha \right) | n\rangle  \right ) d^{2}\alpha   \in {\cal H}
%\label{B1kana}
.  \end{align}% with some complex-valued function and assume $\phi \in {\cal H}$. 
Then, it holds that \begin{align}
 \int_{D(r)} \left(  \sum ^{\infty }_{n=0} a _{n} \alpha ^n | n\rangle \right)  d^{2}\alpha  =\sum_{n=0} ^{\infty} \left(\int_{D(r)} a_{n}   \alpha  ^n  d^{2}\alpha \right) | n\rangle. \label{B22kana0}
\end{align}
\end{theorem}

\begin{remark}
The area $D$ is not necessary in the form of the disk.  We merely use the condition %that $D \subset \mathbb C $ is a compact domain, i.e.,  
$|D |= \int_D d^2 \alpha  < \infty$. 
\end{remark}

\begin{remark}
For the verification of the relation in Eq.~\eqref{Qn22B2}, one may proceed to define the form of the state family as $  \ket {\varphi _\alpha  } : =  \sum ^{\infty }_{n=0} e ^{-|\alpha| ^2 /2 } a _{n} \alpha ^n | n\rangle $ instead of the form in Eq.~\eqref{statefamilyvf}. 
\end{remark}

%\begin{remark}
%This is a rather direct result of the dominated convergence theorem for Bochner integral.  
%\end{remark}
%\begin{remark}
%$D$ is not necessary in the form of the disk.  We merely use the condition that $D \subset \mathbb C $ is a compact domain, i.e.,  $|D |= \int_D d^2 \alpha  < \infty$. 
%\end{remark}

%

\begin{proof}

Let us define 
\begin{align}
 M_n: = |a_n| ^2 |r|^{2n}, \\  g_n (\alpha ): = \sum_{k=0}^n |a_k|^2 |\alpha |^{2k} ,\quad  f (\alpha ): = \lim_{n \to \infty }  g_n (\alpha ).
\end{align}
Let us note that the condition (i) implies  $\sum_n M_n $ is convergent and that the following inequality holds  
\begin{align}
 |g_n(\alpha ) -  g_m (\alpha ) | \le \sum_{k = n+1} ^m M_k. 
\end{align}  
This means % not only
 the sequence of functions  $(g_n)_n$ % is point-wise cauchy but also that it 
 converges to $f $ uniformly, namely, 
\begin{align}
 (\forall \epsilon >0 ) \exists N> 0; (\forall n \ge N) \sup_{\alpha \in D} |f(\alpha ) - g_n (\alpha ) |< \epsilon.  
\end{align}  Therefore, we can exchange the order  of the integration and infinite summation as there exists  a sufficiently large $N >0 $ such that for $n \ge N$ it holds
\begin{align}
& \left | \int_D f(\alpha ) d^2\alpha -   \sum_{k=0}^n \int_D |a_k|^2 \left | \alpha \right |^{2k} d^2\alpha     \right | \nonumber \\
= & \left | \int_D f(\alpha ) d^2\alpha -  \int_D g_n (\alpha ) d^2\alpha     \right | \nonumber \\   \le & \int_D  \left | f(\alpha ) -   g_n (\alpha ) \right |d^2\alpha < |D| \ \epsilon ,    
\end{align} where we use the fact that a finite summation and an integration are commutable due to the linearity of integrals in the first line, and the first inequality in the final line is due to the triangle inequality for integrals.  % and in the last line we define $ |D|:= \int_D d^2 \alpha    $ is finite.  
Thus far  we have proven
\begin{align}
 \int_D \sum_{k=0}^\infty  |a_k|^2  \left | \alpha \right |^{2k} d^2\alpha  = \sum_{k=0}^\infty  \int_D |a_k|^2  \left | \alpha \right |^{2k} d^2\alpha < \infty ,  \label{QnThusFar}
\end{align}
where the finiteness is due to the condition (ii). This is nothing more than the term-wise integrability of a power series. We will associate this relation to the square summable property in the number space ${\cal H}$. % so as to consider the set of quantum states whose coefficients in the number basis are  convergent power series. 

Let us remind that  a convergent power series defines a continuous function. This imples the following family of state vectors 
\begin{align}
  \ket {\varphi _\alpha  } : =  \sum ^{\infty }_{n=0} a _{n} \alpha ^n | n\rangle  \in {\cal H} \label{statefamilyvf}
\end{align}
is continuous with respect to $\alpha \in D(r) $, that is, it holds   
\begin{align}
 \| \ket {\varphi_ \alpha}  - \ket {\varphi_\beta}  \| \to 0  \qquad \left (|\alpha - \beta | \to 0 \right ).  
\end{align}
Therefore, its integral over the area $D$ is well-defined (as the limit of a Riemann sum): 
\begin{align}
 \phi :=&\int_{D(r)}  \ket{\varphi_\alpha }d ^2 \alpha =  \int_{D(r)} \left(  \sum ^{\infty }_{n=0} a _{n} \alpha ^n | n\rangle \right)  d^{2}\alpha. 
 %\int_{D}   \sum _{n = 0} ^\infty   \varphi _{n}\left( \alpha \right) \ket{n} \ d^{2} \alpha. 
\end{align}
Similarly, integrals of truncated states in the following form exist, 
\begin{align}
 \phi^{(N-1)} :=& \int_{D}   \sum _{n = 0} ^N  a _{n} \alpha ^n   \ket{n} \ d^{2} \alpha. \label{QnFinSumPhi}%,  \\
%\psi^{(N)} :=& \sum  _{n= 0}^N \left(\int_{D }  a _{n}  \alpha ^n  d^{2}\alpha \right) \ket{n}.
\end{align}
By using the triangle inequality for integrals and Schwartz's inequality,  we obtain  
\begin{align}
\| \phi-  \phi^{(N)} \| = & \left \|  \int_{D(r)} \left(  \sum ^{\infty }_{n=N} a _{n} \alpha ^n | n\rangle \right)  d^{2}\alpha \right \| \nonumber \\
\le & \int_{D(r)} \left\|  \sum ^{\infty }_{n=N} a _{n} \alpha ^n | n\rangle \right\|  d^{2}\alpha \nonumber \\
=&  \int_D \left (  { \sum ^{\infty }_{n=N} \left | a _{n} \right |^2 \left | \alpha \right| ^{2n}   } \right ) ^{1/2}d^2 \alpha  \nonumber \\
\le &  |D|^{1/2}  \left ( \int_D  { \sum ^{\infty }_{n=N} \left | a _{n} \right |^2 \left | \alpha \right| ^{2n}   }  d^2 \alpha.
\right )^{1/2}\end{align}
Since the integral  in the last expression vanishes as $N \to \infty $ due to Eq.~\eqref{QnThusFar}, the sequence of states $(\phi^{(N)}) $ converges to $\phi$: 
\begin{align}
\| \phi-  \phi^{(N)} \|  \to 0 \quad (N \to \infty ). \label{QnPhiCau}
\end{align}

%\begin{align}
%\end{align}
%
%
%
%Let us define 
%\begin{align}
% \phi^{(N)} :=& \int_{D}   \sum _{n = 0} ^N  a _{n} \alpha ^n   \ket{n} \ d^{2} \alpha,  \\
%\psi^{(N)} :=& \sum  _{n= 0}^N \left(\int_{D }  a _{n}  \alpha ^n  d^{2}\alpha \right) \ket{n}.
%\end{align}
%Since the summations are   finite  it holds
%\begin{align}
% \| \phi^{(N)} - \psi^{(N)} \| =0. 
%\end{align}
%

Next, let us define 
\begin{align}
 \mathbb C \ni b_n := \int _D  a_{n}   \alpha ^{n }  d^{2}\alpha.  
\end{align}
We can readily show that the sequence $(b_n)$ is square-summable as follows: 
 Due to Schwartz's inequality it holds 
\begin{align}
 |b_n| = &  \left|\int_{D }  a_{n}  \alpha ^n   d^{2}\alpha \right| % \nonumber  \\ 
  \le   \int_{D } \left| a_{n}  \alpha ^n  \right|  d^{2}\alpha  \nonumber  \\
  \le &  |D|^{1/2} \sqrt{ \int_{D } \left| a_{n}  \alpha ^n  \right|^2  d^{2}\alpha}. 
\end{align} Then, use of Eq.~\eqref{QnThusFar} yields 
\begin{align}
 \sum_n |b_n|^2  
  \le & |D |  \sum_n   %\sqrt
  { \int_{D } \left| a_{n}    \alpha ^n  \right|^2  d^{2}\alpha} < \infty.    
\end{align} % where we use to ensure the finiteness.  
% Eq.~ @@ \eqref{T4inQ1} in the final inequality. 
Therefore, the state vector in the form of  
\begin{align}
 \psi : =\sum_{n=0}^\infty b_n \ket{n} =& \sum  _{n= 0}^\infty  \left(\int_{D }  a _{n}  \alpha ^n  d^{2}\alpha \right) \ket{n},  %\in {\cal H}
\end{align} exists in $\cal H $ as well as its truncated states
\begin{align}
 \psi^{(N-1)} : =\sum_{n=0}^N b_n \ket{n} =& \sum  _{n= 0}^N \left(\int_{D }  a _{n}  \alpha ^n  d^{2}\alpha \right) \ket{n}.  \label{QnFinSumPsi}  %\in {\cal H}
\end{align}
Obviously,   $(\psi ^{(N)})$  defines a Cauchy sequence converges to  $\psi$  in ${\cal H}$,   \begin{align}
 \|\psi - \psi^{(N)}\| \to 0 \qquad (N \to \infty). \label{QnPsiCau}
\end{align}
In turn, another obvious fact is  $\| \phi^{(N)} - \psi^{(N)} \| =0$ as  the summations in Eq.~\eqref{QnFinSumPhi} and Eq.~\eqref{QnFinSumPsi} are finite. %This implies  $( \phi^{(N)})$ is also cauchy, and its limit $\phi $ has to satisfy
%\begin{align}
% \|\phi - \phi^{(N)}\| \to 0 \qquad (N \to \infty). 
%\end{align}

Finally combining Eqs.~\eqref{QnPhiCau} and \eqref{QnPsiCau} with  the following  triangular inequality 
\begin{align*}
\left \| \phi - \psi \right\| = & \left \|  \phi- \phi^{(N)} +  \phi^{(N)} - \psi^{(N)} + \psi^{(N)} - \psi \right \| \nonumber \\
\le & \left \|  \phi- \phi^{(N)} \right \| +   \left \|  \phi^{(N)} - \psi^{(N)} \right \|+   \left \|  \psi^{(N)} - \psi \right \|
\nonumber \\
 =& \left \|  \phi- \phi^{(N)} \right \|+   \left \|  \psi^{(N)} - \psi \right \|,  %\nonumber \\
\end{align*}
we obtain 
\begin{align}
 \left \| \phi - \psi \right\| =\left \|  \phi- \phi^{(N)} \right \|+   \left \|  \psi^{(N)} - \psi \right \|  \to 0  \qquad (N \to \infty). 
\end{align} 
This relation implies the conclusion of our theorem
\begin{align}
 \int_{D}   \sum _{n = 0} ^\infty  \varphi _{n}\left( \alpha \right) \ket{n} \ d^{2} \alpha  = \sum  _{n= 0}^\infty \left(\int_{D }  \varphi _{n}\left( \alpha \right)  d^{2}\alpha \right) \ket{n}. 
\end{align}
\end{proof}


\begin{thebibliography}{8}%
\makeatletter
\providecommand \@ifxundefined [1]{%
 \@ifx{#1\undefined}
}%
\providecommand \@ifnum [1]{%
 \ifnum #1\expandafter \@firstoftwo
 \else \expandafter \@secondoftwo
 \fi
}%
\providecommand \@ifx [1]{%
 \ifx #1\expandafter \@firstoftwo
 \else \expandafter \@secondoftwo
 \fi
}%
\providecommand \natexlab [1]{#1}%
\providecommand \enquote  [1]{``#1''}%
\providecommand \bibnamefont  [1]{#1}%
\providecommand \bibfnamefont [1]{#1}%
\providecommand \citenamefont [1]{#1}%
\providecommand \href@noop [0]{\@secondoftwo}%
\providecommand \href [0]{\begingroup \@sanitize@url \@href}%
\providecommand \@href[1]{\@@startlink{#1}\@@href}%
\providecommand \@@href[1]{\endgroup#1\@@endlink}%
\providecommand \@sanitize@url [0]{\catcode `\\12\catcode `\$12\catcode
  `\&12\catcode `\#12\catcode `\^12\catcode `\_12\catcode `\%12\relax}%
\providecommand \@@startlink[1]{}%
\providecommand \@@endlink[0]{}%
\providecommand \url  [0]{\begingroup\@sanitize@url \@url }%
\providecommand \@url [1]{\endgroup\@href {#1}{\urlprefix }}%
\providecommand \urlprefix  [0]{URL }%
\providecommand \Eprint [0]{\href }%
\providecommand \doibase [0]{https://doi.org/}%
\providecommand \selectlanguage [0]{\@gobble}%
\providecommand \bibinfo  [0]{\@secondoftwo}%
\providecommand \bibfield  [0]{\@secondoftwo}%
\providecommand \translation [1]{[#1]}%
\providecommand \BibitemOpen [0]{}%
\providecommand \bibitemStop [0]{}%
\providecommand \bibitemNoStop [0]{.\EOS\space}%
\providecommand \EOS [0]{\spacefactor3000\relax}%
\providecommand \BibitemShut  [1]{\csname bibitem#1\endcsname}%
\let\auto@bib@innerbib\@empty
%</preamble>
\bibitem [{\citenamefont {Twareque~Ali}\ \emph {et~al.}(1995)\citenamefont
  {Twareque~Ali}, \citenamefont {Antoine}, \citenamefont {Gazeau},\ and\
  \citenamefont {Mueller}}]{CSReview95}%
  \BibitemOpen
  \bibfield  {author} {\bibinfo {author} {\bibfnamefont {S.}~\bibnamefont
  {Twareque~Ali}}, \bibinfo {author} {\bibfnamefont {J.-P.}\ \bibnamefont
  {Antoine}}, \bibinfo {author} {\bibfnamefont {J.-P.}\ \bibnamefont
  {Gazeau}},\ and\ \bibinfo {author} {\bibfnamefont {U.}~\bibnamefont
  {Mueller}},\ }\bibfield  {title} {\bibinfo {title} {Coherent states and their
  generalizations: a mathematical overview},\ }\href@noop {} {\bibfield
  {journal} {\bibinfo  {journal} {Reviews in Mathematical Physics}\ }\textbf
  {\bibinfo {volume} {7}},\ \bibinfo {pages} {1013} (\bibinfo {year}
  {1995})}\BibitemShut {NoStop}%
\bibitem [{\citenamefont {Antoine}(2015)}]{Antoine2000}%
  \BibitemOpen
  \bibfield  {author} {\bibinfo {author} {\bibfnamefont {J.-P.}\ \bibnamefont
  {Antoine}},\ }\bibfield  {title} {\bibinfo {title} {Coherent states and
  wavelets, a contemporary panorama},\ }in\ \href@noop {} {\emph {\bibinfo
  {booktitle} {Operator Algebras and Mathematical Physics}}},\ \bibinfo
  {editor} {edited by\ \bibinfo {editor} {\bibfnamefont {T.}~\bibnamefont
  {Bhattacharyya}}\ and\ \bibinfo {editor} {\bibfnamefont {M.~A.}\ \bibnamefont
  {Dritschel}}}\ (\bibinfo  {publisher} {Springer International Publishing},\
  \bibinfo {address} {Cham},\ \bibinfo {year} {2015})\ pp.\ \bibinfo {pages}
  {123--156}\BibitemShut {NoStop}%
\bibitem [{\citenamefont {Garcia~de Leon}(2008)}]{GarciadeLeon08}%
  \BibitemOpen
  \bibfield  {author} {\bibinfo {author} {\bibfnamefont {P.~L.}\ \bibnamefont
  {Garcia~de Leon}},\ }\bibfield  {title} {\bibinfo {title} {Coherent state
  quantization for conjugated variables},\ }\href@noop {} {\bibfield  {journal}
  \href{https://tel.archives-ouvertes.fr/tel-00432055}
  {\bibinfo  {journal} {NNT: 2008PEST0210 tel-00432055}\ } (\bibinfo {year}
  {2008})}\BibitemShut {NoStop}%
\bibitem [{\citenamefont {Parisio}(2010)}]{Parisio2010}%
  \BibitemOpen
  \bibfield  {author} {\bibinfo {author} {\bibfnamefont {F.}~\bibnamefont
  {Parisio}},\ }\bibfield  {title} {\bibinfo {title} {{Off-Center
  Coherent-State Representation and an Application to Semiclassics}},\ }\href
  {https://doi.org/10.1143/PTP.124.53} {\bibfield  {journal} {\bibinfo
  {journal} {Prog. Theor. Phys.}\ }\textbf {\bibinfo {volume} {124}},\ \bibinfo
  {pages} {53} (\bibinfo {year} {2010})}\BibitemShut {NoStop}%
\bibitem [{\citenamefont {Klauder}()}]{KlauderLectureNote}%
  \BibitemOpen
  \bibfield  {author} {\bibinfo {author} {\bibfnamefont {J.~R.}\ \bibnamefont
  {Klauder}}, Problem 1.5, \ }\href {https://www.phys.ufl.edu/~klauder/norway/} {\emph
  {\bibinfo {title} {2006 Norway lectures,
  https://www.phys.ufl.edu/~klauder/norway/}}}\BibitemShut {NoStop}%
\bibitem [{\citenamefont {Christensen}(2001)}]{Christensen2001}%
  \BibitemOpen
  \bibfield  {author} {\bibinfo {author} {\bibfnamefont {O.}~\bibnamefont
  {Christensen}},\ }\bibfield  {title} {\bibinfo {title} {Frame, riesz~bases,
  and discrete gabor/wavelet expansions},\ }\href@noop {} {\bibfield  {journal}
  {\bibinfo  {journal} {Bull.Amer.Math.Soc.}\ }\textbf {\bibinfo {volume}
  {38}},\ \bibinfo {pages} {273} (\bibinfo {year} {2001})}\BibitemShut
  {NoStop}%
\bibitem [{\citenamefont {Yoshida}(1980)}]{FAYoshida}%
  \BibitemOpen
  \bibfield  {author} {\bibinfo {author} {\bibfnamefont {K.}~\bibnamefont
  {Yoshida}},\ }\href@noop {} {\emph {\bibinfo {title} {Functional analysis}}}\
  (\bibinfo  {publisher} {Springer-Verlag,Berlin},\ \bibinfo {year}
  {1980})\BibitemShut {NoStop}%
%\bibitem [{\citenamefont {Lin}\ and\ \citenamefont
% {L\"utkenhaus}(2020)}]{PhysRevApplied.14.064030}%
%  \BibitemOpen
%  \bibfield  {author} {\bibinfo {author} {\bibfnamefont {J.}~\bibnamefont
%  {Lin}}\ and\ \bibinfo {author} {\bibfnamefont {N.}~\bibnamefont
%  {L\"utkenhaus}},\ }\bibfield  {title} {\bibinfo {title} {Trusted detector
%  noise analysis for discrete modulation schemes of continuous-variable quantum
%  key distribution},\ }\href {https://doi.org/10.1103/PhysRevApplied.14.064030}
%  {\bibfield  {journal} {\bibinfo  {journal} {Phys. Rev. Applied}\ }\textbf
%  {\bibinfo {volume} {14}},\ \bibinfo {pages} {064030} (\bibinfo {year}
%  {2020})}\BibitemShut {NoStop}%
\end{thebibliography}
\end{document}